\documentclass[a4paper,UKenglish,cleveref,autoref,thm-restate]{lipics-v2021}
\pdfoutput=1 %
\hideLIPIcs  %

\usepackage{amsmath}
\usepackage{amssymb}
\usepackage{multirow}
\usepackage{algorithm}
\usepackage[noEnd=true,indLines=true]{algpseudocodex}
\usepackage{booktabs,subcaption,amsfonts}
\usepackage{xspace}
\usetikzlibrary{plotmarks}
\usepackage{pgfplots}
\pgfplotsset{compat=1.18}

\newcommand{\set}[1]{\left\{%
      \begin{array}{l}#1\end{array}
      \right\}}
\newcommand{\sset}[2]{\left\{~#1  \left|
      \begin{array}{l}#2\end{array}
    \right.     \right\}}
\newcommand{\etal}{\textit{et al.\ }}
\newcommand{\false}{\mathit{false}}
\newcommand{\true}{\mathit{true}}
\newcommand\tuple[1]{\langle #1 \rangle}
\newcommand{\pats}{\textsf{pats}}
\newcommand{\pat}{\textsf{pat}}
\newcommand{\dmc}{\texttt{d4}\xspace}
\newcommand{\AllPats}{\textit{AllPats}}
\newcommand{\DomPats}{\textit{DomPats}}

\title{Breaking Symmetries with Involutions} %

\author{Michael Codish}{Department of Computer Science, Ben-Gurion
  University of the Negev, Israel}{mcodish@bgu.ac.il}{https://orcid.org/0000-0003-0394-5854}{}
\author{Mikol{\'{a}}\v{s} Janota}{Czech Technical University in Prague, CIIRC, Czechia}{mikolas.janota@gmail.com}{https://orcid.org/0000-0003-3487-784X}{}

\authorrunning{M. Codish and M. Janota} %

\Copyright{Michael Codish and Mikol{\'{a}}\v{s} Janota} %



\category{} %

\relatedversion{} %

\acknowledgements{%
The authors would like to thank the anonymous reviewers for their
valuable comments and constructive feedback, which helped improve the
quality and clarity of this paper.
This work began at the \emph{Dagstuhl-Seminar 23261 SAT Encodings and Beyond}.
The results were supported by the Ministry of Education, Youth and Sports within
the dedicated program ERC~CZ under the project \emph{POSTMAN} no.~LL1902 and
co-funded by the European Union under the project \emph{ROBOPROX}
(reg.~no.~CZ.02.01.01/00/22\_008/0004590).
}

\nolinenumbers %

\EventEditors{Maria Garcia de la Banda}
\EventNoEds{1}
\EventLongTitle{31st International Conference on Principles and Practice of Constraint Programming (CP 2025)}
\EventShortTitle{CP 2025}
\EventAcronym{CP}
\EventYear{2025}
\EventDate{August 10--15, 2025}
\EventLocation{Glasgow, Scotland}
\EventLogo{}
\SeriesVolume{340}
\ArticleNo{25}

\begin{document}
\maketitle

\begin{abstract}
  Symmetry breaking for graphs and other combinatorial objects is
  notoriously hard. On the one hand, complete symmetry breaks are
  exponential in size. On the other hand, current, state-of-the-art,
  partial symmetry breaks are often considered too weak to be of
  practical use.
  Recently, the concept of graph patterns has been introduced and
  provides a concise representation for (large) sets of non-canonical
  graphs, i.e.\ graphs that are not lex-leaders and can be excluded
  from search.  In particular, four (specific) graph patterns apply to
  identify about 3/4 of the set of all non-canonical graphs.
  Taking this approach further, we discover that graph patterns that
  derive from permutations that are involutions play an important
  role in the construction of symmetry breaks for graphs.
  We take advantage of this to guide the construction of partial and
  complete symmetry-breaking constraints based on graph patterns.
  The resulting constraints are small in size and strong in the number
  of symmetries they break.
\end{abstract}

\section{Introduction}

Graph search problems are about finding simple graphs with desired
structural properties. Such problems arise in many real-world
applications and are fundamental in graph theory.
Solving graph search problems is typically hard due to the enormous
search space and the large number of symmetries in graph
representation: Any graph obtained by permuting the vertices of a
solution (or a non-solution) is also a solution (or a non-solution),
and is considered isomorphic, or ``symmetric''.
The set of all such isomorphic graphs forms an isomorphism class.
To optimize search, we aim to restrict it to focus on one
``canonical'' graph from each isomorphism class.

One common approach to eliminate symmetries is to add \emph{symmetry
  breaking constraints} that are satisfied by at least one member of
each isomorphism class~\cite{crawford-kr96,Shlyakhter07,Walsh06}.
A symmetry-breaking constraint is called \emph{complete} if it is
satisfied by exactly one member of each isomorphism class and
\emph{partial} otherwise.
In many cases, symmetry-breaking constraints, complete or partial, are
expressed in terms of ``lex-constraints''. Each lex-constraint
corresponds to one symmetry, $\sigma$, which is a permutation on
vertices, and restricts the search space to consider assignments that
are lexicographically smaller than their permuted form obtained
according to $\sigma$.
If one considers the set of lex-constraints corresponding to all
$n$-factorial permutations, then the corresponding symmetry break is
complete but too large to be of practical use.

Itzhakov and Codish~\cite{Itzhakov2016} observe that a complete
symmetry break can be defined in terms of a number of lex-constraints
which is considerably smaller than $n$-factorial. They succeed in
computing complete symmetry-breaking constraints of practical size for
graphs with 10 or less vertices.  Dan\v{c}o~et~al.\ obtain similar
results in the context of finite models~\cite{danco2025}.  But this
approach does not
scale. Itzhakov~\etal\cite{DBLP:conf/flops/CodishEGIS18} introduce
the notion of ``lex-implications'', which are a refinement of
lex-constraints. In the case of graphs, they compute complete
symmetry breaks which are much smaller in size. But still, this approach
does not scale beyond graphs with 11 vertices.
It is known that breaking symmetry by adding constraints to eliminate
symmetric solutions is intractable in
general~\cite{Babai1983,crawford-kr96}.  So, we do not expect to find
a complete symmetry break of polynomial size that identifies canonical
graphs that are lex-leaders.
Hence, the interest in partial symmetry breaks.

Codish~\etal\cite{Codish2019} introduce a partial symmetry break,
which is equivalent to considering the quadratic number of
permutations that swap a pair of vertices
(transpositions). Rintanen~\etal\cite{RintanenR24} enhance this
approach for directed graphs.  This approach is widely applied and
turns out to work well in practice, despite eliminating only a small
portion of the symmetries.  However, when dealing with hard instances
of graph search problems, this constraint does not suffice.
Over the past decade, there has been little progress in the research
of partial symmetry-breaking constraints for graph search
problems. Some attempts are made
in~\cite{DBLP:conf/cpaior/ItzhakovC23} and
in~\cite{DBLP:conf/cp/CodishGIS16}. However, there are no references in
the literature to the applications that make use of the symmetry
breaks defined in these papers.

In a recent paper, Codish and Janota~\cite{cpaior2025}
introduce a set-covering perspective on symmetry breaking.
A permutation ``covers'' a graph if its application on the graph
yields a smaller graph with respect to a given (lexicographic)
order. A complete symmetry break is then a set of permutations that
covers all non-canonical graphs.
In that paper, the authors introduce the notion of a ``pattern'' which
provides a concise representation for the set of graphs covered by a
permutation.

In this paper, we pick up on the notion of ``patterns'' which we call
here ``graph patterns'' and show how they can be used to provide
complete and also partial symmetry breaks which improve on the current
state-of-the-art.
Motivated by the set-covering perspective presented
in~\cite{cpaior2025}, we implement a greedy algorithm to cover the set
of all non-canonical graphs using graph patterns. Although this
approach does not scale, it provides a test bed to study the structure
of the graph patterns selected in the greedy approach.
We observe that the first four graph patterns selected in the greedy
approach already cover $3/4$ of the set of all graphs (and slightly
more of the set of all non-canonical graphs).
Our study also reveals the importance of a specific type of
permutation when breaking symmetries on graphs. This is the class of
permutations that are equal to their inverse, also called
\emph{involutions}. This class includes the classes of transpositions
(which swap a pair of indices) and of consecutive transpositions
(which swap a pair of consecutive vertices) that are at the basis of
the partial symmetry breaks introduced
in~\cite{DBLP:conf/ijcai/CodishMPS13,Codish2019}.

We implement an algorithm that guides the search for a complete
symmetry break based on the structure of graph patterns. Along the
way, we obtain a chain of partial symmetry breaks of increasing
precision. We present experimental results to evaluate these symmetry
breaks.
These results show that in the setting of a counterexample-guided
abstraction-based (CEGAR) addition of symmetry breaks, focusing on
``involutions first'' reduces the number of iterations of the
CEGAR-loop. Broadly speaking, this shows that it is beneficial to look
for symmetry-breaking constraints systematically, rather than
arbitrarily.

The rest of the paper is structured as follows.
Section~\ref{sec:prelim} provides definitions and notations that
are used throughout. A series of examples are provided to
demonstrate the concepts upon which we build.
Section~\ref{sec:greedy} builds on the set-covering approach to
symmetry breaking and presents a greedy approach to cover all of the
non-canonical graphs with graph patterns. While this approach does not
scale, it provides a test bed with which to study the structure of the
(``best'') graph patterns selected.
Section~\ref{sec:CegarTweak} applies the lessons learned from the
greedy approach to guide a CEGAR-based algorithm to collect specific
types of graph patterns that derive from various types of involutions.
Section~\ref{sec:experiments} describes a series of experiments in
which we provide complete and partial symmetry breaks to a specific
graph search problem and evaluate their performance.
Finally, Section~\ref{sec:conclusions} concludes and presents future work.

\section{Preliminaries and Notation}\label{sec:prelim}

\noindent\textbf{Representing Graphs:~~}
Throughout this paper, we consider simple graphs, i.e.\ undirected
graphs with no self-loops.
The adjacency matrix of a graph $G$ is an $n\times n$ Boolean
matrix. The element at row $i$ and column $j$ is $\true$ if and only
if $(i,j)$ is an edge. The list of edges of a graph $G$ is denoted
$edges(G)$. It consists of the $\binom{n}{2}$ elements obtained as
the concatenation of the columns of the upper triangle of $G$ (or any
other predetermined order).
In abuse of notation, we let $G$ denote a graph in
any of its representations.
An \emph{unknown graph} of order-$n$ is represented as an $n\times n$
adjacency matrix of Boolean variables which is symmetric and has the
values $\false$ (denoted by 0) on the diagonal, or as the
corresponding list of edges.

\medskip\noindent\textbf{Ordering Graphs:~~}
Let $G_1,G_2$ be known or unknown graphs with $n$ vertices. Then,
$G_1 \leq G_2$ if and only if $edges(G_1)\leq_{lex} edges(G_2)$ where
$\leq_{lex}$ denotes the standard lexicographic ordering.
When $G_1$ and $G_2$ are unknown graphs, then the lexicographic
ordering, $G_1\leq G_2$, can be viewed as specifying a
\emph{lexicographic order constraint} over the variables in $G_1$
and $G_2$.
For Boolean strings $\bar a=\tuple{a_1,\ldots, a_m}$ and
$\bar b=\tuple{b_1, \ldots, b_m}$ and $1\leq i\leq m$, we denote
  \begin{equation}
    \label{eq:smaller-i}
    \bar a<_{lex}^i \bar b \Leftrightarrow
    \set{(a_1=b_1),\ldots, (a_{i-1}=b_{i-1}),
      (a_i=0), (b_i=1)}
  \end{equation}
We denote $G_1<^i G_2$ if $edges(G_1)<_{lex}^i edges(G_2)$.

\medskip\noindent\textbf{Permutations:~~}
The group of permutations on $\{1 \ldots n\}$ is denoted $S_n$.
We represent a permutation $\pi \in S_n$ as a sequence of length $n$
where the i$^{th}$ element indicates the value of $\pi(i)$.  For
example: the permutation $[2,3,1] \in S_3$ maps as follows:
$\set{1 \mapsto 2, 2 \mapsto 3, 3 \mapsto 1}$.
We will refer to the following types of permutations:

\medskip\noindent\textbf{transpositions:~~}
A \emph{transposition} is a permutation that swaps two elements and
leaves all other in place. For example, $[4,2,3,1]$ is the
transposition that swaps elements $1$ and $4$. A transposition that
swaps a pair of consecutive elements, is called a \emph{consecutive
  transposition}. For example, $[1,3,2,4]$ is a consecutive
transposition.

\medskip\noindent\textbf{involutions:~~}
An \emph{involution} is a permutation $\pi$ such that $\pi\circ\pi$ is
the identity. Involutions are permutations that swap any set of pairs
of disjoint elements. For example, $[4,3,2,1]$ is the involution which
swaps the pairs $\set{1,4}$ and $\set{2,3}$.  Transpositions are a
special case of involutions.  A \emph{consecutive involution} is an
involution that swaps any number of consecutive pairs. For example
$[2,1,4,3]$ is a consecutive involution. We say that an involution is
\emph{intersecting} if it swaps (among others) the pairs $\set{i,j}$
and $\set{k,l}$ with $i<j$ and $k<l$ such that the intervals $[i,j]$
and $[k,l]$ are intersecting. If an involution is not intersecting,
then we say that it is \emph{disjoint}. For example, $[3,4,1,2]$ (in
cyclic notation $(1~ 3)(2~ 4)$) is a disjoint involution because the
intervals $[1,3]$ and $[2,4]$ are intersecting.
The number of involutions is given as the sequence \texttt{A000085} of
  the Online Encyclopedia of Integer Sequences.
Note that an involution is always a composition of finitely many
disjoint transpositions.

\medskip

Permutations act on graphs and on unknown graphs in the natural
way. For a graph and also for an unknown graph $G$,
viewing $G$ as an adjacency matrix, given a permutation $\pi\in S_n$,
then $\pi(G)$ is the adjacency matrix obtained by mapping each element
at position $(i,j)$ to position $(\pi(i),\pi(j))$ (for $1\leq i,j\leq
n$). The composition of permutations is defined in the natural way.
Two graphs $G,H$ are \emph{isomorphic} if there exists a
permutation $\pi \in S_n$ such that $G=\pi(H)$.
Note that applying an involution to the vertices of a graph (directed or
undirected) induces an involution on its edges because
$(\pi\circ\pi(i),\pi\circ\pi(j))=(i,j)$ for any involution $\pi$.

\medskip\noindent\textbf{Symmetry Breaks:~~}
A \emph{symmetry break} for graph search problems is a predicate, $\psi(G)$,
on a graph $G$, which is satisfied by at least one graph in each
isomorphism class of graphs.
If $\psi$ is satisfied by exactly one graph in each isomorphism class,
then we say that $\psi$ is a \emph{complete symmetry break}. Otherwise, it is
\emph{partial}.

Heule~\cite{Heule2019} defines the notion of \emph{redundancy ratio}, which we denote
$\rho(\psi)$, to measure the precision of $\psi$.  This is the ratio
between the number of graphs that satisfy $\psi(G)$ and the number of
isomorphism classes.  One can view $\rho(\psi)$ as the average number
of graphs per isomorphism class that are not eliminated by $\psi$.

In our setting, the canonical graphs, are the minimal (lexleader)
graphs of the isomorphism classes of graphs. The following example is
adapted from~\cite{cpaior2025} and demonstrates some of the concepts
introduced so far. Notice that, in this paper, edge variables are
ordered by columns.

\begin{example}\label{example:lex-constraints}
  The following depicts an unknown, order-4, graph $G$, its
  permutation $\pi(G)$, for $\pi = [1,2,4,3]$, and their
  representations as lists of edges.  The lex-constraint
  $G \leq \pi(G)$ can be simplified as described by
  Frisch~\etal\cite{Frisch03} to:
  $\tuple{x_2,x_3} \leq_{lex} \tuple{x_4,x_5}$.

  \noindent
  \resizebox{\textwidth}{!}{%
  {\small\begin{tabular}{p{.253\textwidth}p{.285\textwidth}l}
 {$\mathbf{G=}\left[\begin{matrix}
    0   & x_1 & x_2 & x_4 \\
    x_1 & 0   & x_3 & x_5 \\
    x_2 & x_3 & 0   & x_6 \\
    x_4 & x_5 & x_6 & 0 
\end{matrix}\right]$}
&
  {$\mathbf{\pi(G)=}\left[\begin{matrix}
    0 & x_1 & x_4 & x_2 \\
    x_1 & 0 & x_5 & x_3 \\
    x_4 & x_5 & 0 & x_6 \\
    x_2 & x_3 & x_6 & 0 
\end{matrix}\right]$}
&
 $\begin{array}{l}
   \text{edges}(G)=\tuple{x_1,x_2,x_3,x_4,x_5,x_6}\\  
   \text{edges}(\pi(G))=\tuple{x_1,x_4,x_5,x_2,x_3,x_6}
\end{array}$
\end{tabular}
}}

  There are 64 graphs of order 4. Eleven of these are canonical and
  the other 53 are non-canonical. The canonical graphs, represented as
  lists of edges, are detailed below.

  {\small
\begin{verbatim}
[0,0,0,0,0,0] [0,0,0,0,0,1] [0,0,0,0,1,1] [0,0,0,1,1,1] [0,0,1,0,1,1] [0,0,1,1,0,0] 
[0,0,1,1,0,1] [0,0,1,1,1,1] [0,1,1,1,1,0] [0,1,1,1,1,1] [1,1,1,1,1,1] 
\end{verbatim}
}
\end{example}

\medskip\noindent\textbf{Graph Patterns:~~}
A \emph{graph pattern} is a partially instantiated graph $G$ (some elements
are variables) such that all instances of $G$ are non-canonical. We
typically represent  a graph pattern $G$ using the list notation $\text{edges}(G)$.

\begin{example}\label{graph patterns order 4}
  The following six graph patterns describe all of the 53 non-canonical graphs
  of order~4.  
  {\small
\begin{verbatim} 
      [1,0,A,B,C,D] [A,1,0,B,C,D] [A,1,B,0,C,D] 
      [A,B,1,B,0,C] [A,A,B,C,1,0] [A,B,B,1,0,C]
\end{verbatim}
  }
\noindent
So, an order 4 graph is canonical if and only if it
is not an instance of any  of these six graph patterns. One can check
that this is the case  for the 11 canonical graphs detailed in
Example~\ref{example:lex-constraints}. 
\end{example}

For demonstration, it is easy to see why the first graph pattern of
Example~\ref{graph patterns order 4} is a valid pattern. If the vector
representation of the graph $G$ starts with 1,0, which correspond to
the edges $\{1,2\}, \{1,3\}$, these can be swapped by the
transposition $[1,3,2,4]$ (swap vertices 2 and 3) resulting in a smaller
graph (this will possibly affect other edges, but these are no longer
important for the lexicographic comparison). So, $G$ is not canonical.

Formally, graph patterns derive from permutations as stated in the
following definition which is adapted from~\cite{cpaior2025} where
edge variables are ordered by rows (in this paper we order the edges
by columns).

\begin{definition}[\cite{cpaior2025}]\label{def:patterns}
  Let $\pi$ be a permutation, $G$ be an unknown graph of order-$n$
  with $edges(G)=\tuple{x_1,\ldots,x_m}$,
  $edges(\pi(G))= \tuple{y_1,\ldots,y_m}$, and let $1\leq i\leq m$.
  The graph pattern, $\pat_i(\pi)$, is the result of applying the most
  general unifier to the equations from the right side in
  Equation~\eqref{eq:smaller-i} for the case of
  $edges(\pi(G))<_{lex}^i edges(G)$ (i.e., the most general
  substitution that makes both sides of the equations identical).
  If the equations have no solution, then we denote
  $\pats_i(\pi)=\bot$.
  The set of all graph patterns of order $n$ is denoted:
  $\AllPats(n) = \sset{\pat_i(\pi)\neq\bot}{\pi\in S_n, 1\leq i\leq
\binom{n}{2}}$.
\end{definition}

The graph pattern $\pat_i(\pi)$ represents the set of graphs that get
smaller at position $i$ with $\pi$. If $\pi$ is a particular type of
permutation, e.g., a transposition, then we say that  $\pat_i(\pi)$ is
a pattern of that type.
Note that $\pat_i(\pi)$ is a legal
graph pattern because graphs that get smaller under some permutation
are trivially non-canonical.
Observe that $\pat_3([1,3,2,4])=\bot$, as in this case
$edges(G)=[x_1,x_2,x_3,x_4,x_5,x_6]$, $edges(\pi(G))=[x_2,x_1,x_3,x_4,x_6,x_5]$
and $edges(\pi(G))<_{lex}^3 edges(G)$ has no solution (because $x_3\not < x_3$).
Observe also, that in some cases, different permutations lead to
identical patterns. For example,
$\pat_3([2,5,1,3,4]) = \pat_3([2,5,1,4,3]) =
[x_1,x_1,1,x_2,x_3,x_4,0,x_1,x_5,x_6]$.

\begin{example}\label{smart_cover-1}
The following details how the graph patterns from Example~\ref{graph
  patterns order 4} are derived from permutations.
  \[
\begin{array}{ll}
    \pat_1([1,3,2,4])=[1,0,A,B,C,D] &\hspace{2cm}
    \pat_2([2,1,3,4])=[A,1,0,B,C,D]\\
    \pat_2([1,2,4,3])=[A,1,B,0,C,D]&\hspace{2cm}
    \pat_3([1,2,4,3])=[A,B,1,B,0,C]\\
    \pat_4([2,1,3,4])=[A,B,B,1,0,C]&\hspace{2cm}
    \pat_5([1,3,2,4])=[A,A,B,C,1,0]\\
  \end{array}
\]

\end{example}

\begin{definition}[cover]\label{def:cover}
  Let $\pi\in S_n$, $1\leq i\leq \binom{n}{2}$, and let
  $\pat_i(\pi)\neq\bot$. 
  Then, $\pat_i(\pi)$ covers a graph $G$ if $G$ is an
  instance of $\pat_i(\pi)$.
  Equivalently, $\pat_i(\pi)$ covers $G$ if 
  $\pi(G)<^i G$.
  We denote the set $cover(\pat_i(\pi)) = \sset{G}{\pi(G)<^i G}$.
  Sometimes we write also $cover(\pi,i)$. If $\pat_i(\pi)=\bot$ then
  it has no instances and so in this case, we define
  $cover(\pat_i(\pi)) = \emptyset$. The number of graphs covered by a
  graph pattern $p$ is $2^k$ where $k$ is the number of variables in $p$.
\end{definition}

Some graph patterns ``dominate''  others. 
\begin{definition}[dominate]
  We say that graph pattern $p_1$ dominates graph pattern $p_2$ if
  $cover(p_2)\subseteq cover(p_1)$. Given a set $S$ of graph patterns,
  $dominators(S)$ denotes the set of dominating atoms in $S$. We
  denote $\DomPats(n)=\text{dominators}(\AllPats(n))$.
\end{definition}

Viewing graph patterns as first-order logic terms, $p_1$ dominates $p_2$ if $p_1$ is more
general than $p_2$. In Prolog, one can implement a test for ``$p_1$
dominates $p_2$'' using the built-in operator
$\mathtt{subsumes\_term(p_1,p_2)}$.

\begin{example}
  There are 59 elements in $\AllPats(4)$ (after removing
  redundancies). The set $\DomPats(4)$ contains only 18 elements.
  For instance,
  
  \begin{tabular}{l}
    $\pat_3([1,2,4,3])=[A,B,1,B,0,C]$
    ~~dominates ~~
    $\pat_3([1,4,2,3])=[A,A,1,A,0,B]$, and\\
    $\pat_2([2,1,4,3])=[A,1,B,C,0,D]$
    ~~dominates ~~
    $\pat_3([3,4,2,1])=[A,1,1,B,0,A]$.
  \end{tabular}
  
  \medskip\noindent
   
\end{example}

Table~\ref{table:dom} details the total number of graph patterns and
the number of dominating graph patterns (in parentheses) for various
types of permutations with small values of $n$. For consecutive
transpositions, and for transpositions, all of the patterns are
dominating. For other types of involutions, a large majority of the
patterns are dominating.

\begin{table}[t]\begin{center}
  \caption{Numbers of graph patterns (and dominating graph patterns) for
  various types of permutations: consecutive transpositions
  \texttt{(ct)}, transpositions \texttt{(t)}, consecutive involutions
  with transpositions \texttt{(ci+t)}, disjoint involutions \texttt{(di)},
  involutions \texttt{(i)}, and all permutations \texttt{(all)}
}
  \begin{tabular}{|c||c|c|c|c|c|c|}
    \hline
    order & \multicolumn{1}{c|}{ct} & \multicolumn{1}{c|}{t}
    & \multicolumn{1}{c|}{ci+t} & \multicolumn{1}{c|}{di}
    & \multicolumn{1}{c|}{i} & \multicolumn{1}{c|}{all} \\
    \hline
    4 & 6 ~(6)  & 12 ~(12)  & 14 ~(14)  & 14 ~(14)  & 17  ~(16)   & 59~(18) \\
    5 & 12~(12) & 30 ~(30)  & 40 ~(39)  & 47 ~(46)  & 80  ~(75)   & 550~(163)\\
    6 & 20~(20) & 60 ~(60)  & 92 ~(88)  & 136~(130) & 348 ~(327)  & 4610~(1648)\\
    7 & 30~(30) & 105~(105) & 187~(176) & 361~(339) & 1451~(1369) & 43065~(17945)\\
    8 & 42~(42) & 168~(168) & 354~(329) & 906~(842) & 6055~(5762) & 421435~(199509)\\
    \hline
  \end{tabular}
  \label{table:dom}
\end{center}\end{table}

Some graph patterns are orthogonal to others.

\begin{definition}[orthogonal]\label{def:orthogonal}
  Let $p_1$ and $p_2$ be (non-$\bot$) graphs patterns. 
  We say that  $p_1$ is orthogonal to  $p_2$
  if $cover(p_1)\cap cover(p_2)=\emptyset$.
\end{definition}

In Prolog, one can implement a test for ``$p_1$
is orthogonal to $p_2$'' using the built-in operator for ``not
unifiable''.

\medskip\noindent\textbf{CEGAR:~~}
In~\cite{Itzhakov2016} and in~\cite{DBLP:conf/flops/CodishEGIS18} the
authors compute complete symmetry breaks for graphs based respectively
on permutations and on a refinement of permutations which they term
``implications''. A similar approach is applied in~\cite{danco2025} to
break the symmetries for finite models.
Common to all of these works is an algorithm based on
\emph{counter-example guided abstraction refinement (CEGAR)}~\cite{cegar2000}.

In a nutshell, and in its simplest form, the CEGAR-based algorithm
performs as follows: Let $\Psi$ denote a set of
permutations, which is initially empty. The algorithm repeatedly seeks
a counter-example to the statement: ``\emph{$\Psi$ is a complete
  symmetry break}''. A counter-example takes the form: graph $G$ and
permutation $\pi$ such that
\[ \pi(G)<G \wedge \bigwedge_{\pi'\in\Psi}G\leq \pi'(G).\] If such a
counter-example is found then $\Psi=\Psi\cup\{\pi\}$. If no such
counter-example is found, then $\Psi$ is a complete symmetry break.
The search for a counter-example is implemented using a SAT encoding
and incremental SAT solving.\footnote{
  The SAT solver cadical-2.1.0~\cite{cadical} was used in all our experiments.
}

An important detail is that $\Psi$ may contain redundant elements. For
example, if a permutation added at some point becomes redundant in
view of permutations added later. A second phase of the algorithm
iterates on the elements of $\Psi$ to remove redundant permutations
(similar to the iterative algorithm for a minimally unsatisfiable set or
monotone predicates in general~\cite{humus,monotone}). It is important
to note that the time to perform the second phase is costly. For
example, in~\cite{Itzhakov2016}, the authors report that the time to
compute the complete symmetry break $\Psi$ for order $10$ graphs is
close to $10$ hours and the time to reduce it is $84$ hours.

In this paper, we focus on symmetry breaks defined in terms of graph
patterns. We seek a set of graph patterns that covers all of the
non-canonical graphs.  The set $\AllPats(n)$ of all graph patterns
clearly covers all non-canonical graphs and hence it is a complete
symmetry break. But this set is too large to be of practical use.
We adapt the CEGAR approach to graph patterns. In this setting, $\Psi$
is a set of graph patterns, and we repeatedly seek a counter-example:
a graph $G$ and a graph pattern $\pat_i(\pi)$ such that $G$ is covered by
$\pat_i(\pi)$ but is not covered by any of the graph patterns in $\Psi$. If
such a counter-example is found then
$\Psi=\Psi\cup\{\pat_i(\pi)\}$. If no such counter-example is found,
then $\Psi$ is a complete symmetry break.

Let $p=[p_1,\ldots,p_m]$ be a graph pattern and let
$G=[x_1,\ldots,x_m]$ denote an unknown graph. Let $I_1$ denote the set
of induces in $p$ which contain a zero, $I_2$ denote the set
of indices in $p$ which contain a one, and $I_3$ denote the set of
pairs of indices that contain equal variables.
The single clause:
\[ c =
  (\bigvee_{i\in I_1} x_i ) \vee
  (\bigvee_{i\in I_2} \neg x_i ) \vee
  (\bigvee_{(i,j)\in I_3} x_i\mathbin{\oplus} x_j)
\]
encodes that $G$ is not covered by $p$. Strictly speaking, $c$ is not
a clause due to the \texttt{xor} operations. It is straightforward to
replace $x_i\mathbin{\oplus} x_j$ by a fresh variable $x_{ij}$ and to
add clauses for $x_{ij}\leftrightarrow x_i\mathbin{\oplus} x_j$
(note that these fresh variables are reused across
the encoding of different graph patterns).

\begin{example}
  let $p=[A,B,1,B,0,C]$. The clause
  $(x_5 \vee \neg x_3 \vee x_2\mathbin{\oplus}x_4)$ encodes that 
  $G=[x_1,x_2,x_3,x_4,x_5,x_6]$ is not covered by $p$.
\end{example}

In~\cite{DBLP:conf/flops/CodishEGIS18}, the authors present a CEGAR
based approach to derive complete symmetry breaks based on
``lex-implications''. Essentially, lex-implications are constraints of
the form detailed in Equation~\eqref{eq:smaller-i}. Similar to graph
patterns, lex-implications can be represented as pairs of the form
$(\pi,i)$.
The novelty in the presentation in this paper stems from the graph
pattern perspective where it is natural to talk about concepts such as
cover, dominance and orthogonality. In the graph pattern perspective,
any set $S$ of graph patterns is a symmetry break in the sense that a
graph $G$ is canonical, \emph{only if} it is not an instance of
an element in~$S$.
Moreover, $S$ is a complete symmetry break if also the inverse
\emph{(if)} statement holds.

\section{Breaking Symmetries with Graph Patterns: A Greedy Approach}\label{sec:greedy}

The CEGAR approach outlined above breaks symmetries as they are
suggested by the SAT solver as counter-examples.
This might lead to unnecessarily large sets of graph patterns
as the SAT solver makes arbitrary choices relatively to the ultimate objective.
In this section, we apply a greedy approach to derive symmetry breaks
consisting of graph patterns. At the core of the approach is the
notion of ranking a  graph pattern in view of the set $\Psi$ of
graph patterns selected so far.

\begin{definition}\label{def:ranking}
  Let $\Psi$ be a set of graph patterns (a symmetry break) and let $p$
  be a graph pattern. We define $\textrm{ranking}_\Psi(p)$ to be the number of
  graphs covered by $p$ but not covered by any of the elements of
  $\Psi$
\end{definition}

The basic idea is presented as Algorithm~\ref{alg:symBreak}. At
Line~\ref{LineA}, The set $S$ of candidate graph patterns is
initialized to the set of dominating graph patterns. The symmetry break
$\Psi$ is initialized to the empty set. At each step of the algorithm,
a graph pattern $p$ with maximal ranking is selected and the sets $S$ and
$\Psi$ are updated.

\begin{algorithm}[t]
\begin{algorithmic}[1]
  \Procedure{symbreak}{$n$}
  \State $S\gets \DomPats(n)$; $\Psi\gets\emptyset$;\label{LineA}
  \Repeat
  \State select $p\in S$ such that\label{LineB}
  \State \qquad $r= \max_{p'\in S}{\textrm{ranking}_\Psi(p')}$;\label{LineC}
  \State \qquad $\textrm{ranking}_\Psi(p) = r$;\label{LineD}
  \If{r>0}
  \State $S\gets S\setminus\{p\}$; $\Psi\gets\Psi\cup\{p\}$;
  \EndIf
  \Until{$r = 0$};
  \State\Return{$\Psi$}; \Comment{$\Psi$ is a complete symmetry break}
  \EndProcedure
\end{algorithmic}
\caption{Greedily compute symmetry break for order $n$ graphs}%
\label{alg:symBreak}
\end{algorithm}

In the implementation, ranking the current set $S$ of candidate graph
patterns (Line~\ref{LineC} in Algorithm~\ref{alg:symBreak}) is a
bottleneck. Each graph pattern is ranked using a sat encoding with a model
counter. As an optimization, in the loop at
Lines~\ref{LineB}--\ref{LineD}, we remove from $S$ all graph patterns
$p$ such that $\textrm{ranking}_\Psi(p) = 0$. Moreover, whenever we select a
pattern $p\in S$ at Line~\ref{LineB}, we also select a set
$S'\subseteq S$ (as large as possible) so that $\{p\}\cup S'$ are
mutually orthogonal (as prescribed by Definition~\ref{def:orthogonal}).
The selection of $S'$ is greedy: When selecting $p\in S$ at Line~\ref{LineB},
initialize $S'=\{p\}$ and iterate over the patterns from S, adding a
pattern to $S'$ if it is orthogonal to those already in $S'$.

\begin{example}
  The following details the run of Algorithm~\ref{alg:symBreak} for
  $n=4$. The algorithm selects 6 graph patterns in 3 rounds. The
  column labeled $\Delta$ indicates the number of graphs covered by
  this graph pattern and not covered by those above.  In this example,
  the second graph pattern added in each round is orthogonal to that
  from the first added in the round. The selected graph patterns are
  the same as those detailed in Example~\ref{graph patterns order
    4}. Note that the sum of the numbers in the $\Delta$ column is 53,
  corresponding to the number of non-canonical graphs of order 4.

  \medskip
\begin{tabular}{|c|c|c|}
\hline
round & pattern & $\Delta$\\
\hline
1 & [1,0,A,B,C,D] & $16$\\
1 & [A,1,0,B,C,D] & $16$\\
2 & [A,1,B,0,C,D] & $8$\\
2 & [A,B,1,B,0,C] & $6$\\
3 & [A,A,B,C,1,0] & $5$\\
3 & [A,B,B,1,0,C] & $2$\\
\hline
\end{tabular}
\end{example}

The following proposition identifies four specific graph patterns
that cover $3/4$ of the total number of graphs. Observe that the
permutations detailed in the proposition are consecutive
transpositions. 

\begin{proposition}\label{strong-four}
For $n\geq 5$, the four top ranking graph patterns always take the 
following particular form where $m=\binom{n}{2}-2$. Below, we adopt the
cycle notation for permutations. The permutation $(j,k)$ is the
permutation which swaps elements $j$ and $k$ and leaves all other
elements fixed.

\begin{enumerate}
\item  $\pat_1((2,3))=[1,0,x_1,x_2,x_3,x_4,x_5,x_6,\ldots,x_m]$
\item  $\pat_2((1,2))=[x_1,1,0,x_2,x_3,x_4,x_5,x_6,\ldots,x_m]$
\item  $\pat_2((3,4))=[x_1,1,x_2,0,x_3,x_4,x_5,x_6,\ldots,x_m]$
\item  $\pat_4((4,5))=[x_1,x_2,x_3,1,x_4,x_5,0,x_6,\ldots,x_m]$
\end{enumerate}

\noindent
The first two graph patterns are orthogonal and each covers $2^{m-2}$
graphs.  The second two are also orthogonal and each covers $2^{m-3}$
graphs from those not covered by the first two. Together they cover
$3/4$ of the total number of graphs.
\end{proposition}
\begin{proof}(sketch)
  The proof follows from a theorem presented in~\cite{Kvasnicka:1990:CIC:83354.83356},
  which states that if $G$ is
  an order $n$ lex-leader canonical graph where the order of edges
  is column-wise, then so is the subgraph $G(k)$ on the first $k$
  vertices for $1\leq k< n$.
  It follows that for any position $i< \binom{n}{2}$, the graph pattern
  $\pat_i((j,k))$ on order-$n$ graphs is also a graph pattern for
  order-$n+1$ graphs.
\end{proof}

Table~\ref{table:profiling} compares 6 symmetry breaks and profiles
them with respect to 5 types of permutations.
The first three symmetry breaks are obtained using
Algorithm~\ref{alg:symBreak}: \texttt{greedy complete} is the complete
symmetry break, \texttt{greedy partial} is the partial
symmetry break consisting of the first (highest ranking) 50\% of the
patterns obtained from Algorithm~\ref{alg:symBreak}, and \texttt{greedy ct
  partial} is the partial symmetry break consisting of the (longest)
prefix of patterns obtained which are all consecutive transpositions.
The next two symmetry breaks are obtained using the CEGAR algorithm
described in~\cite{DBLP:conf/flops/CodishEGIS18}: \texttt{cegar
  complete} is the complete symmetry break obtained, and \texttt{cegar
  partial} is the partial symmetry break consisting of the first 50\%
of the patterns obtained from the CEGAR algorithm.
We will come back to the fifth symmetry break, \texttt{involutions
  partial} later.

The column labeled \texttt{total} details the total number of
patterns in the symmetry break. The next 5 columns detail the number
of patterns in the symmetry break derived from a permutation of a
given type:
\texttt{(ct)} consecutive transpositions,
\texttt{(t)} all transpositions,
\texttt{(ci+t)} consecutive involutions and transpositions,
\texttt{(di)} disjoint involutions, and
\texttt{(i)} all involutions.
All of these types of permutations are specific forms of
involutions. The first two: \texttt{(ct)} and \texttt{(t)} are 
those widely applied in the symmetry breaks defined in~\cite{DBLP:conf/ijcai/CodishMPS13,Codish2019}.

The penultimate column, labeled $\rho$, shows the redundancy
ratio. This is the ratio between the number of graphs that are not
covered by the selected patterns (symmetries not broken) and the
number of isomorphism classes (all symmetries broken). If the set of
patterns is a complete symmetry break, then $\rho=1$. The smaller this
number is, the better the symmetry break. The last column, labeled
\texttt{\%\;ncc}, details the percentage of non-canonical graphs
covered by the symmetry break. If the set of patterns is a complete
symmetry break, then this number is 100.

  \begin{table}[t]\begin{center}
    \caption{Profiling partial and complete symmetry breaks obtained from greedy algorithm}\label{table:profiling}
      \begin{tabular}{|c|l|r|r|r|r|r|r|r|r|}
        \hline
            n     & symmetry break & total & ct & t & ci+t & di & i  & $\rho$~~~ & \%\,ncc~\\
      \hline
      \multirow{4}{*}{7}
      & greedy complete  & 116   & 28 & 30 & 59 & 60 &  80 &  1.00 & 100.00\\
      & greedy partial   &  58   & 28 & 28 & 48 & 49 &  51 &  1.27 & 99.99\\
      & greedy ct partial&  13   & 13 &  0 &  0 &  0 &   0 & 24.38 & 98.78\\
      \cline{2-10}  
      & cegar complete & 136 & 15 & 16 & 44 & 48 & 62 & 1.00 &100.00\\
      & cegar partial    &  68 & 15 & 16 & 25 & 27 & 30 & 1035.25&51.56\\
      \cline{2-10}  
      & involutions partial&113 &28&30&60&70&113& 1.01 & 99.99\\
        \hline
      \multirow{4}{*}{8}  
       & greedy complete & 439   & 40 & 46 & 120 & 124 & 235 & 1.00 & 100.00\\
       & greedy partial & 219   & 40 & 44 & 115 & 116 &   153 & 1.07 & 99.99\\
       & greedy ct partial& 18  & 18 &  0 &  0 &  0 &   0 &52.63  & 99.76\\
       \cline{2-10}  
       & cegar complete & 492   & 19 & 24 & 104 & 114 & 194 & 1.00 & 100.00\\
       & cegar partial & 246   & 19 & 24 & 63 & 71 &   95 & 9710.21& 44.66\\
       \cline{2-10}  
       & involutions partial& 396 & 40&47&127&158&396& 1.02 & 99.99\\
        \hline
    \end{tabular}
  \end{center}
\end{table}

\medskip\noindent\textbf{What we learn from
  Table~\ref{table:profiling}:} 
\begin{quote}
  The first $4-10$\% of the patterns selected are all derived from
  consecutive transpositions. These patterns alone already cover $>98$\%
  of the non-canonical graphs.
  Moreover, as stated in Proposition~\ref{strong-four}, the first four
  patterns always cover a total of $3/4$ of the total number of graphs
  which is also about 75\% of the total number of non-canonical
  graphs.
  
\end{quote}
\begin{quote}
  About half of the patterns selected in \texttt{greedy complete} are
  involutions; These are highly ranked. More than 60\% are in the top
  half (by ranking) of the patterns.
  In contrast, for $n=7$ less than 5\% of all ($n$-factorial)
  permutations are involutions, and for $n=8$ less than 2\%.
  The number of involutions is given as sequence \texttt{A000085} of
  the Online Encyclopedia of Integer Sequences.
\end{quote}
\begin{quote}
  Greedy selection pays off. Comparing the redundancy ratios for
  \texttt{cegar partial} and for \texttt{greedy partial} shows this.
  Basically, the CEGAR-based algorithm selects an arbitrary pattern
  which covers some graph that is not yet covered. In contrast to the
  greedy algorithm where the best such pattern is chosen.
\end{quote}
\begin{quote}
  There are a small number of graph patterns that cover a large
  portion of the search space. All of these patterns derive from
  consecutive transpositions. This is apparent from the last column in
  the table and the rows corresponding to \texttt{greedy ct partial}:
  for n=7, 13 graph patterns cover $>98\%$ of the non-canonical
  graphs, and for $n=8$,  18 graph patterns cover $>99\%$.
  We propose to consider three numbers when evaluating the quality of
  a symmetry break: (1) the redundancy ratio, (2) the percentage of
  non-canonical graphs covered, and (3) the size of the symmetry
  break. It is easy to obtain perfect values for the first two
  numbers when the third is large.

\end{quote}

Given the apparent importance of various types of involutions, we
consider a fifth symmetry break in Table~\ref{table:profiling}:
\texttt{involutions partial}. Here, we take the set of all dominating
patterns derived from involutions. These are reduced to remove
redundant patterns (a pattern in a set is redundant if all of the
graphs that it covers are covered by other patterns in the set).  For
example, when $n=7$, there are 1369 dominating patterns for
permutations which are involutions (see Table~\ref{table:dom}) and
these can be reduced to 113 patterns after removing redundancies.
As indicated in Table~\ref{table:profiling}, adopting involutions
gives a close-to-perfect symmetry break for small values of~$n$.

The greedy approach does not scale.  Algorithm~\ref{alg:symBreak}
relies on the expensive application of a model counter to rank the
candidate patterns. Also, the computation of all dominating patterns
is too time-consuming.
This motivates the approach taken in the next section.

\section{Tweaking CEGAR for Better Partial and Complete  Symmetry Breaks}\label{sec:CegarTweak}

In this section we take the lessons learned from
Table~\ref{table:profiling} and apply them to guide a CEGAR-based
algorithm to make better selections. We layer the selection of
counter-examples with layers corresponding to the five types of
patterns considered in the profiling of Table~\ref{table:profiling}.
This means that in each layer, each iteration of the CEGAR loop
searches for a graph $G$ and a permutation $\pi$ s.t.\ $\pi(G)<G$
and $\pi$ is the permutation used in that layer.
This yields a graph pattern generated from $\pi$ and covering graph $G$.\footnote{
  Graph patterns generated by a single permutation are
  disjoint, hence only one can cover the graph~$G$.
}

In the first layer \texttt{(ct)}, we seek counter-examples that are
consecutive transpositions. When no further counter-examples of this
type remain, we proceed to layer \texttt{(t)} to collect
counter-examples which are transpositions, and so on for the layers
\texttt{(ci+t)} introducing consecutive involutions, \texttt{(di)} for
disjoint involutions, and \texttt{(i)} for all other involutions. In a
final, sixth layer, we seek arbitrary counter-examples.
To ensure that counter-examples are found of the required type,
we encode in the SAT solver an additional condition in a straightforward fashion.
For example, involutions are encoded as the implications
$\pi(i) = j\Rightarrow\pi(j) =i$, for $i\neq j$.

The layered CEGAR-based algorithm can be applied in two different
capacities: First, to provide partial symmetry breaks as obtained at
the end of each layer of the run. Second, to guide the search for a
complete symmetry break. Experimentation indicates that making a better
selection of counter-examples early on reduces the number of
iterations applied in the course of the algorithm and reduces the
number of redundant patterns that need to be removed in the second
phase of the algorithm.

\noindent
\begin{table}[t]
  \caption{Comparing partial symmetry breaks obtained by layered CEGAR}\label{tab:layered}
  \resizebox{\textwidth}{!}{%
\begin{tabular}{|p{1.3em}|r|r|r|r|r|r|r|r|r|r|r|}
\hline
ord
& \multicolumn{2}{c|}{ct}
& \multicolumn{2}{c|}{t}
& \multicolumn{2}{c|}{ci+t}
& \multicolumn{2}{c|}{di}
& \multicolumn{2}{c|}{i}
& \multicolumn{1}{c|}{compl}
\\ \hline
&size&ratio&size&ratio&size&ratio&size&ratio&size&ratio&size\\
\hline
4  & 6  & 1.00 & 6  & 1.00 & 6 & 1.00 &  6 & 1.00 &  6 & 1.00 &  6  \\
5  & 12 & 1.35 & 13 & 1.26 & 13 & 1.06 &  13 & 1.06 &  14 & 1.00 &  14  \\
6  & 20 & 2.08 & 24 & 1.77 & 29 & 1.16 &  30 & 1.12 &  36 & 1.00 &  36  \\
7  & 30 & 3.87 & 40 & 3.02 & 64 & 1.46 &  70 & 1.31 &  111 & 1.01 &  115  \\
8  & 42 & 7.27 & 62 & 5.39 & 137 & 1.95 &  167 & 1.53 &  397 & 1.02 &  444  \\
9  & 56 & 13.05& 91 & 9.42 & 269 & 2.76 &  401 & 1.83 &  2,024 & 1.03 &  2,760  \\
10 & 72 & 21.53& 128& 15.34& 526 & 3.97 &  1,001 & 2.20 &  12,644 & 1.04 &  24,993  \\
11 & 90 & 33.23& 174& 23.52&1,008& 5.71 &  2,523 & 2.65 &  81,522 & $\approx$1.14 &  289,863  \\
12 &110 & 48.97 &  230 & 34.53 &  1,896 & $\approx$7.39 &  6,275 & $\approx$3.12  & n/a &  n/a & n/a \\
  \hline
\end{tabular}
}
\end{table}

Table~\ref{tab:layered} addresses the first capacity. Each pair of
columns details the partial symmetry break obtained after the
corresponding layer in the revised CEGAR algorithm. For each layer, we
present the size (number of patterns) of the corresponding partial
symmetry break and the corresponding redundancy ratio. For the last
layer (in the last column), the symmetry break is complete and
the redundancy ratio is always~$1.00$.
The first two pairs of columns
\texttt{ct} and \texttt{t} correspond to the symmetry breaks
described in~\cite{DBLP:conf/ijcai/CodishMPS13}. 
To compute the redundancy ratios, the knowledge-compilation-based tool
\dmc~\cite{d4} was applied for model counting.\footnote{
  Other model counters were considered but \dmc gave better performance.
}
In cases where the model counter was not able to compute the number of
graphs we applied the approximate model
counter \texttt{approx}~\cite{approxmodelcounter}. These cases are marked by the
symbol $\approx$.

Now we look at the effect of layering when CEGAR is used to calculate 
a complete symmetry break.
Figure~\ref{fig:cegar} summarizes the comparison of the layered
approach to CEGAR and the standard one.
Figure~\ref{fig:cegar} (on the left), details numbers of iterations
(and patterns).  The two upper curves detail the number of iterations
required to calculate a complete cover of all graphs of order
$n\in\{8..11\}$.
Each iteration of the CEGAR loop adds one pattern to
the cover. However, some patterns might become redundant due to the
addition of a stronger pattern later on. The lower two curves detail
the number of patterns in the irredundant cover obtained by reducing
the output of the CEGAR algorithm (irredundant cover is calculated by
standard means~\cite{monotone}). Note that the size of the irredundant
cover must always be smaller or equal to the number of iterations of
the CEGAR iterations.

The number of iterations is always lower for the layered approach than
for the standard one. This typically also reduces the overall time.
For $n=11$, CEGAR needed 400,016 iterations in the layered approach,
contrasting with 557,279 iterations in the non-layered (standard)
approach. Interestingly, both were reduced to  covers of similar
sizes, 289,863 in the layered approach and 291,888 in the non-layered
one. This indicates that the layered approach guides CEGAR \emph{more
  precisely}.

Figure~\ref{fig:cegar} (on the right) complements this information
with the total CPU time needed to calculate the cover, and it can be
observed that lowering the iterations of the loop also reduces the
total time.

\begin{figure}[tb]
    \begin{subfigure}{.49\textwidth}
\begin{tikzpicture}
\begin{semilogyaxis}[
    width=.9\textwidth,
    height=5.25cm,
    xlabel={order},
    ylabel={iterations/patterns},
    title={CEGAR number of iterations},
    grid=major,
    grid style=densely dotted,
    legend pos=north west,
    legend style={draw=gray!50,draw opacity=0.9,text opacity=1,fill opacity=0.5,font=\tiny\sffamily},
    log basis y=10,
    font=\sffamily,
]
\addplot[
    color=blue,
    mark=o,
     mark options={solid},
    mark size=1pt,
    line width=1pt,
    solid
] coordinates {
    (8, 678)
    (9, 4189)
    (10, 36754)
    (11, 400016)
};
\addplot[
    color=red,
    mark=square*,
     mark options={solid},
    mark size=1pt,
    line width=1pt,
    dashed
] coordinates {
    (8, 3586)
    (9, 12652)
    (10, 82789)
    (11, 557279)
};
\addplot[
    color=green!60!black,
    mark=diamond,
     mark options={solid},
    mark size=1pt,
    line width=1pt,
    dotted
] coordinates {
    (8, 444)
    (9, 2760)
    (10, 24993)
    (11, 289863)
};
\addplot[
    color=orange,
    mark=pentagon,
     mark options={solid},
    mark size=1pt,
    line width=1pt,
    dashdotted
] coordinates {
    (8, 487)
    (9, 2864)
    (10, 25461)
    (11, 291888)
};
\legend{layered, not layered, \#pats after red. - layered, \#pats after red. - not layered}
\end{semilogyaxis}
\end{tikzpicture}
    \end{subfigure}
    \begin{subfigure}{.49\textwidth}
\begin{tikzpicture}
\begin{semilogyaxis}[
    width=.9\textwidth,
    height=5.25cm,
    xlabel={order},
    ylabel={seconds},
    title={CEGAR time comparison},
    grid=major,
    grid style=densely dotted,
    legend pos=north west,
    legend style={draw=gray!50,draw opacity=0.9,text opacity=1,fill opacity=0.5,font=\tiny\sffamily},
    log basis y=10,
    font=\sffamily,
]
\addplot[
    color=blue,
    mark=o,
     mark options={solid},
    mark size=1pt,
    line width=1pt,
    solid
] coordinates {
    (8, 3.462)
    (9, 31.273)
    (10, 649.771)
    (11, 46925.086)
};
\addplot[
    color=red,
    mark=square*,
     mark options={solid},
    mark size=1pt,
    line width=1pt,
    dashed
] coordinates {
    (8, 5.737)
    (9, 54.387)
    (10, 1201.755)
    (11, 83070.297)
};
\legend{layered, not layered}
\end{semilogyaxis}
\end{tikzpicture}
    \end{subfigure}
    \caption{Comparison of CEGAR and Layered CEGAR (log scale)}\label{fig:cegar}
\end{figure}
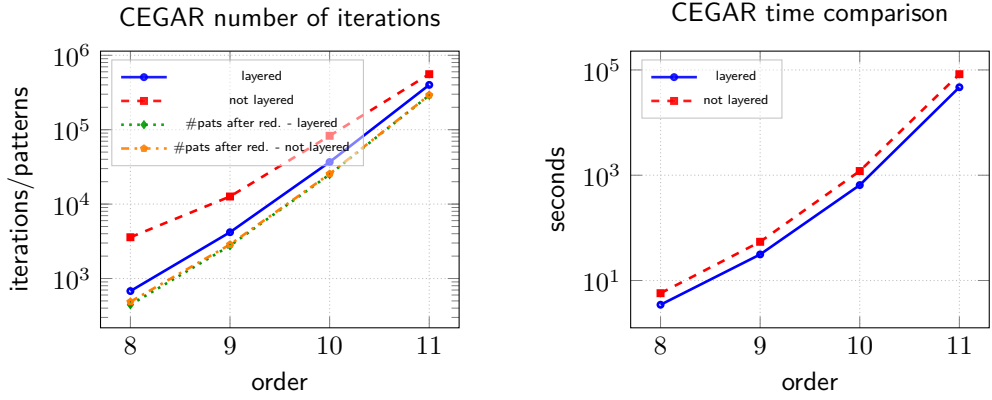

\begin{figure}[h!]
  \begin{center}
\begin{tikzpicture}
\begin{semilogyaxis}[
    width=.9\textwidth,
    height=5.25cm,
    xlabel={phase},
    ylabel={seconds},
    xtick={0,1,2,3,4,5},
    xticklabels={ct,t,ci+t,di,i,all},
    title={layered CEGAR time comparison},
    grid=major,
    grid style=densely dotted,
    legend pos=north west,
    legend style={draw=gray!50,draw opacity=0.9,text opacity=1,fill opacity=0.5,font=\tiny\sffamily},
    log basis y=10,
    font=\sffamily,
]
\addplot[
    color=blue,
    mark=o,
     mark options={solid},
    mark size=1pt,
    line width=1pt,
    solid
] coordinates {
    (0, 0.01)
    (1, 0.03)
    (2, 0.05)
    (3, 0.12)
    (4, 0.78)
    (5, 2.26)
};
\addplot[
    color=red,
    mark=square*,
     mark options={solid},
    mark size=1pt,
    line width=1pt,
    dashed
] coordinates {
    (0, 0.02)
    (1, 0.06)
    (2, 0.13)
    (3, 0.61)
    (4, 6.5)
    (5, 20.47)
};
\addplot[
    color=green!60!black,
    mark=diamond,
     mark options={solid},
    mark size=1pt,
    line width=1pt,
    dotted
] coordinates {
    (0, 0.03)
    (1, 0.12)
    (2, 0.45)
    (3, 2.06)
    (4, 89.04)
    (5, 440.82)
};
\addplot[
    color=orange,
    mark=pentagon,
     mark options={solid},
    mark size=1pt,
    line width=1pt,
    dashdotted
] coordinates {
    (0, 0.05)
    (1, 0.23)
    (2, 1.06)
    (3, 8.38)
    (4, 1781.32)
    (5, 31940.32)
};
\legend{n=8, n=9, n=10, n=11}
\end{semilogyaxis}
\end{tikzpicture}
    \caption{Computation time for partial symmetry breaks with layered
      CEGAR (log scale)}\label{fig:partialcegar}
  \end{center}
\end{figure}
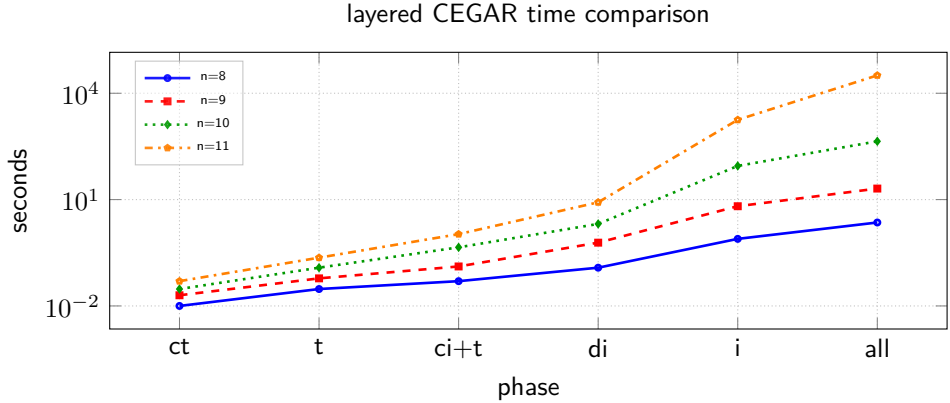

Figure~\ref{fig:partialcegar} details the computation times for
partial symmetry breaks obtained in the different phases of the layered
CEGAR algorithm. For $n=11$, computing the partial symmetry breaks
defined in terms of disjoint involutions and in terms of all
involutions is respectively 3 and 1 orders of magnitude faster than
computing the complete symmetry break. Given the fact that these
symmetry breaks are relatively small and have good redundancy ratios
indicates that they provide a good choice when balancing 
time (to compute them) and size with precision.

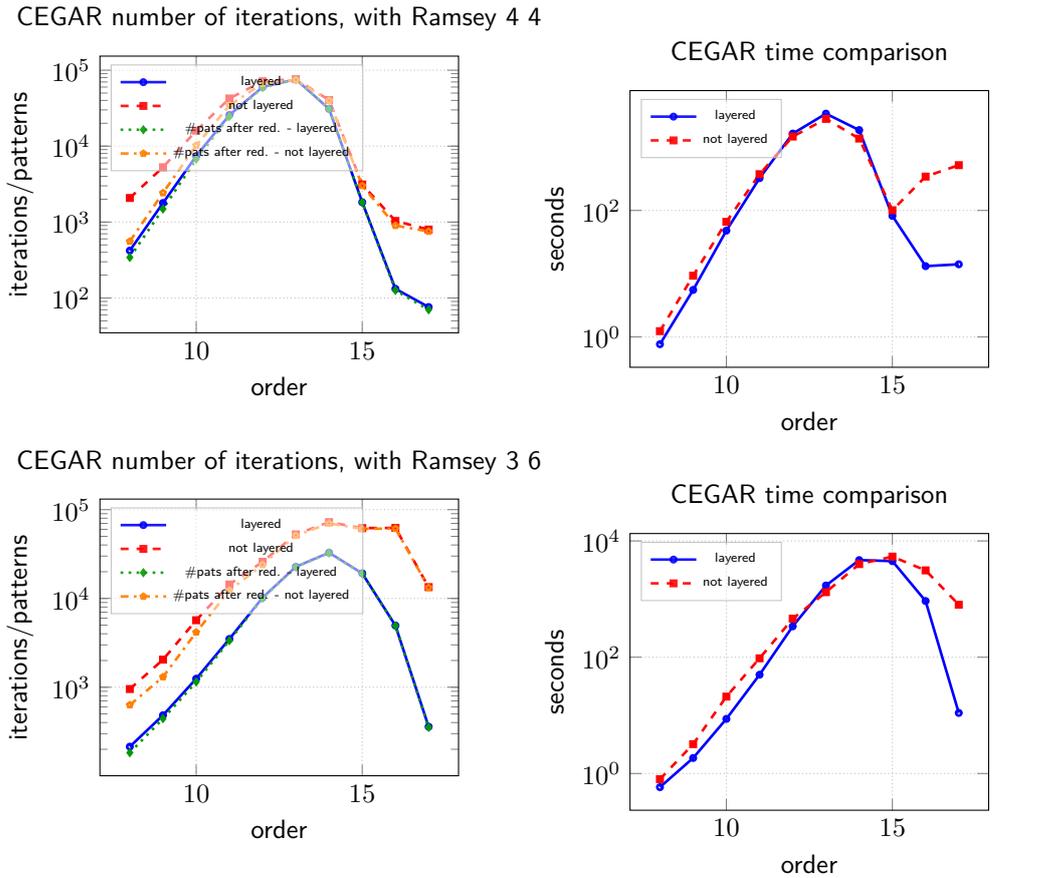
\begin{figure}[th]
    \begin{subfigure}{.49\textwidth}
\begin{tikzpicture}
\begin{semilogyaxis}[
    width=.9\textwidth,
    height=5.25cm,
    xlabel={order},
    ylabel={iterations/patterns},
    title={CEGAR number of iterations, Ramsey 4 4},
    grid=major,
    grid style=densely dotted,
    legend pos=south west,
    legend style={draw=gray!50,draw opacity=0.9,text opacity=1,fill opacity=0.5,font=\tiny\sffamily},
    log basis y=10,
    font=\sffamily,
]
\addplot[
    color=blue,
    mark=o,
     mark options={solid},
    mark size=1pt,
    line width=1pt,
    solid
] coordinates {
    (8, 421)
    (9, 1783)
    (10, 7398)
    (11, 25698)
    (12, 60043)
    (13, 75968)
    (14, 30935)
    (15, 1824)
    (16, 132)
    (17, 76)
};
\addplot[
    color=red,
    mark=square*,
     mark options={solid},
    mark size=1pt,
    line width=1pt,
    dashed
] coordinates {
    (8, 2076)
    (9, 5277)
    (10, 15960)
    (11, 42444)
    (12, 71586)
    (13, 76167)
    (14, 40533)
    (15, 3127)
    (16, 1034)
    (17, 797)
};
\addplot[
    color=green!60!black,
    mark=diamond,
     mark options={solid},
    mark size=1pt,
    line width=1pt,
    dotted
] coordinates {
    (8, 342)
    (9, 1488)
    (10, 6728)
    (11, 24589)
    (12, 59058)
    (13, 75478)
    (14, 30862)
    (15, 1815)
    (16, 126)
    (17, 70)
};
\addplot[
    color=orange,
    mark=pentagon,
     mark options={solid},
    mark size=1pt,
    line width=1pt,
    dashdotted
] coordinates {
    (8, 557)
    (9, 2424)
    (10, 10210)
    (11, 34388)
    (12, 67220)
    (13, 74543)
    (14, 39870)
    (15, 3013)
    (16, 899)
    (17, 752)
};
\legend{layered, not layered, \#pats after red. - layered, \#pats after red. - not layered}
\end{semilogyaxis}
\end{tikzpicture}
    \end{subfigure}
    \begin{subfigure}{.49\textwidth}
\begin{tikzpicture}
\begin{semilogyaxis}[
    width=.9\textwidth,
    height=5.25cm,
    xlabel={order},
    ylabel={seconds},
    title={CEGAR time comparison},
    grid=major,
    grid style=densely dotted,
    legend pos=north west,
    legend style={draw=gray!50,draw opacity=0.9,text opacity=1,fill opacity=0.5,font=\tiny\sffamily},
    log basis y=10,
    font=\sffamily,
]
\addplot[
    color=blue,
    mark=o,
     mark options={solid},
    mark size=1pt,
    line width=1pt,
    solid
] coordinates {
    (8, 0.768)
    (9, 5.54)
    (10, 47.972)
    (11, 324.126)
    (12, 1640.406)
    (13, 3390.329)
    (14, 1865.499)
    (15, 82.032)
    (16, 13.167)
    (17, 14.086)
};
\addplot[
    color=red,
    mark=square*,
     mark options={solid},
    mark size=1pt,
    line width=1pt,
    dashed
] coordinates {
    (8, 1.235)
    (9, 9.297)
    (10, 66.194)
    (11, 374.722)
    (12, 1472.504)
    (13, 2803.278)
    (14, 1367.977)
    (15, 100.828)
    (16, 343.0)
    (17, 520.364)
};
\legend{layered, not layered}
\end{semilogyaxis}
\end{tikzpicture}
 \\ 
    \end{subfigure}
    \begin{subfigure}{.49\textwidth}
\begin{tikzpicture}
\begin{semilogyaxis}[
    width=.9\textwidth,
    height=5.25cm,
    xlabel={order},
    ylabel={iterations/patterns},
    title={CEGAR number of iterations, Ramsey 3 6},
    grid=major,
    grid style=densely dotted,
    legend pos=south west,
    legend style={draw=gray!50,draw opacity=0.9,text opacity=1,fill opacity=0.5,font=\tiny\sffamily},
    log basis y=10,
    font=\sffamily,
]
\addplot[
    color=blue,
    mark=o,
     mark options={solid},
    mark size=1pt,
    line width=1pt,
    solid
] coordinates {
    (8, 214)
    (9, 484)
    (10, 1246)
    (11, 3495)
    (12, 10219)
    (13, 22730)
    (14, 32750)
    (15, 19164)
    (16, 4947)
    (17, 359)
};
\addplot[
    color=red,
    mark=square*,
     mark options={solid},
    mark size=1pt,
    line width=1pt,
    dashed
] coordinates {
    (8, 954)
    (9, 2045)
    (10, 5680)
    (11, 14353)
    (12, 25858)
    (13, 52660)
    (14, 72315)
    (15, 61965)
    (16, 62277)
    (17, 13413)
};
\addplot[
    color=green!60!black,
    mark=diamond,
     mark options={solid},
    mark size=1pt,
    line width=1pt,
    dotted
] coordinates {
    (8, 183)
    (9, 443)
    (10, 1147)
    (11, 3351)
    (12, 10053)
    (13, 22482)
    (14, 32498)
    (15, 19056)
    (16, 4932)
    (17, 353)
};
\addplot[
    color=orange,
    mark=pentagon,
     mark options={solid},
    mark size=1pt,
    line width=1pt,
    dashdotted
] coordinates {
    (8, 631)
    (9, 1305)
    (10, 4157)
    (11, 12296)
    (12, 23959)
    (13, 51074)
    (14, 69982)
    (15, 60762)
    (16, 61221)
    (17, 13234)
};
\legend{layered, not layered, \#pats after red. - layered, \#pats after red. - not layered}
\end{semilogyaxis}
\end{tikzpicture}
    \end{subfigure}
    \begin{subfigure}{.49\textwidth}
\begin{tikzpicture}
\begin{semilogyaxis}[
    width=.9\textwidth,
    height=5.25cm,
    xlabel={order},
    ylabel={seconds},
    title={CEGAR time comparison},
    grid=major,
    grid style=densely dotted,
    legend pos=north west,
    legend style={draw=gray!50,draw opacity=0.9,text opacity=1,fill opacity=0.5,font=\tiny\sffamily},
    log basis y=10,
    font=\sffamily,
]
\addplot[
    color=blue,
    mark=o,
     mark options={solid},
    mark size=1pt,
    line width=1pt,
    solid
] coordinates {
    (8, 0.582)
    (9, 1.849)
    (10, 8.696)
    (11, 50.016)
    (12, 340.424)
    (13, 1712.526)
    (14, 4682.617)
    (15, 4500.523)
    (16, 926.047)
    (17, 11.014)
};
\addplot[
    color=red,
    mark=square*,
     mark options={solid},
    mark size=1pt,
    line width=1pt,
    dashed
] coordinates {
    (8, 0.8)
    (9, 3.2)
    (10, 21.029)
    (11, 95.745)
    (12, 459.917)
    (13, 1314.626)
    (14, 3997.52)
    (15, 5379.873)
    (16, 3117.004)
    (17, 800.658)
};
\legend{layered, not layered}
\end{semilogyaxis}
\end{tikzpicture}
 \\ 
    \end{subfigure}
    \caption{Comparison of CEGAR's Performance on Ramsey graphs
      (log scale)}\label{fig:cegarRamsey}
\end{figure}

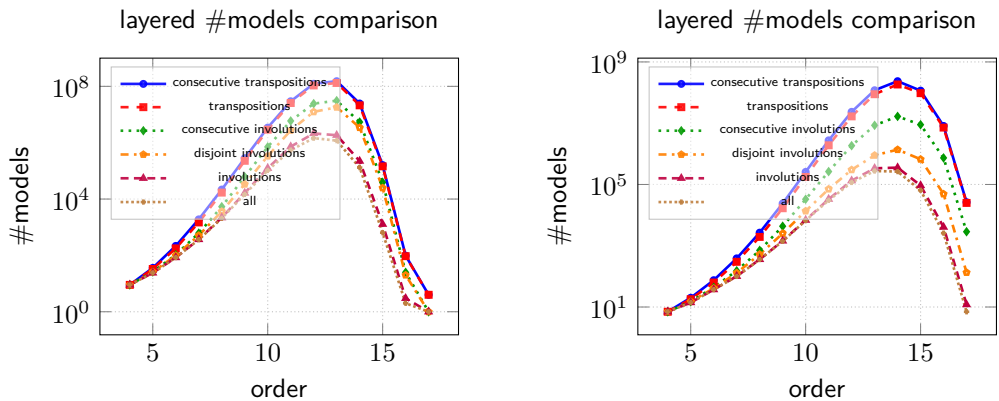
\begin{figure}[th]
    \begin{subfigure}{.495\textwidth}
\begin{tikzpicture}
\begin{semilogyaxis}[
    width=.9\textwidth,
    height=5.25cm,
    xlabel={order},
    ylabel={\#models},
    title={\# models comparison},
    grid=major,
    grid style=densely dotted,
    legend pos=north west,
    legend style={draw=gray!50,draw opacity=0.9,text opacity=1,fill opacity=0.5,font=\tiny\sffamily},
    log basis y=10,
    font=\sffamily,
]
\addplot[
    color=blue,
    mark=o,
     mark options={solid},
    mark size=1pt,
    line width=1pt,
    solid
] coordinates {
    (4, 9)
    (5, 36)
    (6, 217)
    (7, 1901)
    (8, 21703)
    (9, 281743)
    (10, 3424156)
    (11, 29532228)
    (12, 123358542)
    (13, 152319195)
    (14, 24587098)
    (15, 160177)
    (16, 98)
    (17, 4)
};
\addplot[
    color=red,
    mark=square*,
     mark options={solid},
    mark size=1pt,
    line width=1pt,
    dashed
] coordinates {
    (4, 9)
    (5, 33)
    (6, 178)
    (7, 1478)
    (8, 16919)
    (9, 227648)
    (10, 2891024)
    (11, 25616963)
    (12, 107509048)
    (13, 131638650)
    (14, 21181746)
    (15, 144663)
    (16, 94)
    (17, 4)
};
\addplot[
    color=green!60!black,
    mark=diamond,
     mark options={solid},
    mark size=1pt,
    line width=1pt,
    dotted
] coordinates {
    (4, 9)
    (5, 26)
    (6, 105)
    (7, 616)
    (8, 5356)
    (9, 61465)
    (10, 709233)
    (11, 5915712)
    (12, 24330366)
    (13, 30959768)
    (14, 5405563)
    (15, 39159)
    (16, 25)
    (17, 1)
};
\addplot[
    color=orange,
    mark=pentagon,
     mark options={solid},
    mark size=1pt,
    line width=1pt,
    dashdotted
] coordinates {
    (4, 9)
    (5, 26)
    (6, 99)
    (7, 512)
    (8, 3640)
    (9, 33646)
    (10, 333734)
    (11, 2740324)
    (12, 12433939)
    (13, 18073418)
    (14, 3430601)
    (15, 24885)
    (16, 20)
    (17, 1)
};
\addplot[
    color=purple,
    mark=triangle*,
     mark options={solid},
    mark size=1pt,
    line width=1pt,
    densely dashed
] coordinates {
    (4, 9)
    (5, 24)
    (6, 84)
    (7, 368)
    (8, 2199)
    (9, 16395)
    (10, 124065)
    (11, 724228)
    (12, 2132211)
    (13, 1841774)
    (14, 222468)
    (15, 1281)
    (16, 3)
    (17, 1)
};
\addplot[
    color=brown,
    mark=+,
     mark options={solid},
    mark size=1pt,
    line width=1pt,
    densely dotted
] coordinates {
    (4, 9)
    (5, 24)
    (6, 84)
    (7, 362)
    (8, 2079)
    (9, 14701)
    (10, 103706)
    (11, 546356)
    (12, 1449166)
    (13, 1184231)
    (14, 130816)
    (15, 640)
    (16, 2)
    (17, 1)
};
\legend{consecutive transpositions, transpositions, consecutive involutions, disjoint involutions, involutions, all}
\end{semilogyaxis}
\end{tikzpicture}
    \end{subfigure}
    \begin{subfigure}{.49\textwidth}
\begin{tikzpicture}
\begin{semilogyaxis}[
    width=.9\textwidth,
    height=5.25cm,
    xlabel={order},
    ylabel={\#models},
    title={\# models comparison},
    grid=major,
    grid style=densely dotted,
    legend pos=north west,
    legend style={draw=gray!50,draw opacity=0.9,text opacity=1,fill opacity=0.5,font=\tiny\sffamily},
    log basis y=10,
    font=\sffamily,
]
\addplot[
    color=blue,
    mark=o,
     mark options={solid},
    mark size=1pt,
    line width=1pt,
    solid
] coordinates {
    (4, 7)
    (5, 20)
    (6, 75)
    (7, 387)
    (8, 2673)
    (9, 23994)
    (10, 257088)
    (11, 2753125)
    (12, 23470236)
    (13, 119008325)
    (14, 235230065)
    (15, 115591639)
    (16, 8131954)
    (17, 26322)
};
\addplot[
    color=red,
    mark=square*,
     mark options={solid},
    mark size=1pt,
    line width=1pt,
    dashed
] coordinates {
    (4, 7)
    (5, 18)
    (6, 63)
    (7, 299)
    (8, 1943)
    (9, 16755)
    (10, 176047)
    (11, 1899806)
    (12, 16696459)
    (13, 88297504)
    (14, 183469010)
    (15, 95743806)
    (16, 7189964)
    (17, 25005)
};
\addplot[
    color=green!60!black,
    mark=diamond,
     mark options={solid},
    mark size=1pt,
    line width=1pt,
    dotted
] coordinates {
    (4, 7)
    (5, 15)
    (6, 44)
    (7, 152)
    (8, 707)
    (9, 4333)
    (10, 32971)
    (11, 263699)
    (12, 1846565)
    (13, 8520806)
    (14, 16662407)
    (15, 8924445)
    (16, 735984)
    (17, 2868)
};
\addplot[
    color=orange,
    mark=pentagon,
     mark options={solid},
    mark size=1pt,
    line width=1pt,
    dashdotted
] coordinates {
    (4, 7)
    (5, 15)
    (6, 41)
    (7, 130)
    (8, 524)
    (9, 2513)
    (10, 13587)
    (11, 70699)
    (12, 301804)
    (13, 893855)
    (14, 1371025)
    (15, 665199)
    (16, 48879)
    (17, 134)
};
\addplot[
    color=purple,
    mark=triangle*,
     mark options={solid},
    mark size=1pt,
    line width=1pt,
    densely dashed
] coordinates {
    (4, 7)
    (5, 14)
    (6, 37)
    (7, 101)
    (8, 358)
    (9, 1440)
    (10, 6879)
    (11, 33134)
    (12, 133014)
    (13, 336488)
    (14, 354861)
    (15, 92761)
    (16, 4054)
    (17, 12)
};
\addplot[
    color=brown,
    mark=+,
     mark options={solid},
    mark size=1pt,
    line width=1pt,
    densely dotted
] coordinates {
    (4, 7)
    (5, 14)
    (6, 37)
    (7, 100)
    (8, 356)
    (9, 1407)
    (10, 6657)
    (11, 30395)
    (12, 116792)
    (13, 275086)
    (14, 263520)
    (15, 64732)
    (16, 2576)
    (17, 7)
};
\legend{consecutive transpositions, transpositions, consecutive involutions, disjoint involutions, involutions, all}
\end{semilogyaxis}
\end{tikzpicture}
    \end{subfigure}

    \caption{The impact of various partial symmetry breaks on Ramsey
      graphs (log scale)}\label{fig:impactPartial}
\end{figure}

\section{Symmetry Breaks for a Specific Graph Search
  Problem}\label{sec:experiments}

In this section, we describe the computation of partial and complete
symmetry breaks tailored for a specific graph search problem. In this
context, when seeking a counter-example in the CEGAR loop, we restrict
the search to graphs that satisfy the constraints of the given search
problem. We apply this approach for the layered approach to CEGAR and
the standard one.

We consider a classic example: the search for \emph{Ramsey
graphs}~\cite{Rad}. The graph $R(s,t;n)$ is a simple graph with~$n$
vertices that contains no clique of size~$s$, and no independent set
of size~$t$. The \emph{Ramsey number}~$R(s,t)$ is the smallest number~$n$ for
which there is no $R(s,t;n)$ graph.
We also describe the impact of using these symmetry breaks.

Figure~\ref{fig:cegarRamsey} summarizes the comparison of the
computation of symmetry breaks in this context.
We focus on symmetry breaks for $R(4,4;n)$ (two upper plots) and
$R(3,6;n)$ (two lower plots).\footnote{%
  These enable us to perform a precise analysis, however,
  unsatisfiability alone for larger Ramsey numbers is famously
  difficult~\cite{GauthierB24}.}
We study only the values of $n$ until the respective Ramsey numbers,
which are both~18. The number of Ramsey graphs peaks at some point, at
which it is also difficult to calculate the cover.
The two plots on the left detail the numbers of iterations in the
layered/non-layered approaches as well as the numbers of patterns
removed after the CEGAR loops (upper for $Ramsey(4,4;n)$ and lower for
$Ramsey(3,6;n)$). The two plots on the right detail total CPU times
(upper for $Ramsey(4,4;n)$ and lower for $Ramsey(3,6;n)$).

The number of iterations is always lower for the layered approach than
for the standard one. Also, the number of patterns removed in the
reduce phase is lower.
This typically reduces the overall time. This is, however, not always
the case as individual SAT calls are more expensive.

Here, in comparison to Figure~\ref{fig:cegar}, the number of patterns
removed after the CEGAR loop is much smaller.
Interestingly, in the case of $R(3,6;n)$, not only that the number of
iterations is significantly smaller in the layered approach, but it
also remains so after reduction. This suggests that the two runs of
CEGAR+reduction reach completely different local optima.

Figure~\ref{fig:impactPartial} depicts the number of Ramsey graphs
found when applying the various proposed symmetry breaks:
$Ramsey(4,4;n)$ on the left and $Ramsey(3,6;n)$ on the right.  The
curves in both plots from highest to lowest, correspond precisely and
in order to the 6 layers: \texttt{(ct)}, \texttt{(t)},
\texttt{(ci+t)}, \texttt{(di)}, \texttt{(i)}, and \texttt{(all)} (the
complete symmetry break).
The two upper curves describe the application of the widely applied
partial symmetry breaks defined
in~\cite{DBLP:conf/ijcai/CodishMPS13}. The lowest curve describes the
application of a complete symmetry break.
General transpositions (\texttt{t}) do not give much advantage over
just consecutive transpositions \texttt{(ct)}. However, both are by
orders of magnitude worse than the other classes of permutations.
Note the proximity of the curve for all involutions (\texttt{i}) to
the lowest curve (the complete symmetry break). This suggests that
focusing only on involutions gives a symmetry break very close to the
complete break.

\section{Conclusions and Future Work}\label{sec:conclusions}

The objective of this paper is to approach the search for graph symmetry breaks
in a systematic way.
We study the structure of graph patterns that are selected in a greedy
algorithm to break all symmetries for graphs of small orders.
We learn that graph patterns that derive from a specific class of
permutations, called involutions, play an important role.
Involutions generalize the class of transpositions, which play an
important role in breaking symmetries for graphs.
We then refine a CEGAR-based approach to compute symmetry breaks
introducing a layered approach. In this way, we can guide the
CEGAR-based search for a complete symmetry break and also provide a
series of partial symmetry breaks of increasing precision.

It is important to reflect on why we construct partial symmetry breaks
in a CEGAR-based algorithm, in contrast to simply collecting all
permutations of the restricted types. For patterns based on
transpositions, the straightforward construction is doable. However,
for patterns based on the other types of involutions, there
are too many of them and the CEGAR-based approach enables us to select
those that contribute to the corresponding symmetry breaks.

We do not expect to find a complete symmetry break of polynomial size
that identifies canonical graphs that are lex-leaders,
cf.~\cite{KatsirelosNW10}.  However, we still aim to find small,
perhaps even polynomial in size, \emph{partial} symmetry breaks that
break a majority of the symmetries on graphs.
Coming back to the greedy algorithm presented in
Section~\ref{sec:greedy}, for $n=8$, 439 graph patterns are found to
cover all of the 268,423,110 non-canonical graphs of order 8.  The
last 163 graph patterns selected in the greedy algorithm, cover a
negligible total of 382 of these 268,423,110 graphs.
In much the same way that Proposition~\ref{strong-four} points to four
specific graph patterns that cover 75\% of the non-canonical graphs,
the \emph{holy grail} for symmetry breaking for graphs is to specify a small
symmetry break that breaks a significant portion of the symmetries.

This paper opens a number of avenues for future work.
Since focusing on specific permutation types positively impacts CEGAR,
we plan to also apply the lessons learned in this paper in
the context of a dynamic symmetry breaking for graph generation. In this
approach, for example as performed in~\cite{sms}, symmetries are
detected and broken during the generation (and enumeration) of the
solutions of a given graph search problem.
Another direction of research is to identify further sub-classes of permutations
by either refining the framework proposed here or looking for completely different
classes. Such division could also be problem-specific.
A more theoretical direction of research is to find
a justification for why involutions lead to better symmetry breaks.

\end{document}